\documentclass[11pt,manuscript]{acmart}

\usepackage{balance}
\def\BibTeX{{\rm B\kern-.05em{\sc i\kern-.025em b}\kern-.08emT\kern-.1667em\lower.7ex\hbox{E}\kern-.125emX}}

\usepackage{booktabs} 
\usepackage{algorithm, algorithmicx}
\usepackage[noend]{algpseudocode}

\usepackage{hyperref}

\newcommand{{\rh}}{{\widehat r}}
\newcommand{{\Rh}}{{\widehat R}}

\newcommand{\cA}{{\mathcal A}}

\usepackage{mathtools}

\newcommand{\anis}[1]{{\color{blue}\bf [Anisur: #1]}}

\newboolean{short}
\setboolean{short}{false}
\newcommand{\shortOnly}[1]{\ifthenelse{\boolean{short}}{#1}{}}
\newcommand{\onlyShort}[1]{\ifthenelse{\boolean{short}}{#1}{}}
\newcommand{\longOnly}[1]{\ifthenelse{\boolean{short}}{}{#1}}
\newcommand{\onlyLong}[1]{\ifthenelse{\boolean{short}}{}{#1}}
\newcommand{\shortLong}[2]{\ifthenelse{\boolean{short}}{#2}{#1}}
\newcommand{\longShort}[2]{\ifthenelse{\boolean{short}}{#2}{#1}} 
\onlyShort{
\usepackage[font=footnotesize]{caption}
\let\labelindent\relax
\usepackage{enumitem}
\setitemize{noitemsep,topsep=0pt,parsep=0pt,partopsep=0pt}
\setenumerate{noitemsep,topsep=0pt,parsep=0pt,partopsep=0pt}
}

\onlyLong{
\usepackage[font=footnotesize,skip=0pt]{caption}
\settopmatter{printacmref=false} 
\renewcommand\footnotetextcopyrightpermission[1]{} 
\pagestyle{plain} 
}

\makeatletter
\onlyShort{
\algrenewcommand\ALG@beginalgorithmic{\footnotesize}
}
\makeatother

\onlyShort{
\def\acmBooktitle#1{\gdef\@acmBooktitle{#1}}
\setcopyright{rightsretained}
}
\begin{document}

%
\onlyLong{
  \title{Efficient Distributed Algorithms for the $K$-Nearest Neighbors Problem}
}

\shortOnly{
\title{Brief Announcement:
Efficient Distributed Algorithms for the $K$-Nearest Neighbors Problem}
}

\author{Reza Fathi}
\affiliation{%
  \institution{University of Houston}
\city{Houston}
\state{Texas, USA.} \postcode{77204}
}
 \email{rfathi@uh.edu}
%
\author{Anisur Rahaman Molla}
\authornote{Research supported by DST Inspire Faculty research grant DST/INSPIRE/04/2015/002801.}
\affiliation{%
  \institution{Indian Statistical Institute}
\city{Kolkata} 
\state{India}
\postcode{700108}
}
\email{molla@isical.ac.in}
%
\author{Gopal Pandurangan}
\authornote{Supported, in part, by NSF grants IIS-1633720, CCF-1540512,   and CCF-1717075, and by BSF grant 2016419.}
\affiliation{%
\institution{University of Houston}
\city{Houston}
\state{Texas, USA.} 
\postcode{77204}
}
\email{gopalpandurangan@gmail.com}

%

\begin{abstract}
The $K$-nearest neighbors is a basic problem in machine learning with numerous applications.
In this problem, given a (training) set of $n$ data points with labels  and a query point $q$, we want to assign a label to $q$ based on the labels of the
$K$-nearest points to the query. We study this problem
in the {\em $k$-machine model},\footnote{Note that parameter $k$ stands for the number of machines in the $k$-machine model and is independent of $K$-nearest points.} a model for distributed large-scale data.  In this model, we assume that the $n$ points are distributed  among the $k$ machines and  the goal is to  compute an answer given a query point to a machine using a small number of communication rounds.

Our main result is a randomized algorithm in the $k$-machine model that runs in $O(\log K)$ communication rounds with high success probability (regardless of the number of machines $k$ and the number of points $n$). The message complexity of the algorithm
is small taking only $O(k\log K)$ messages. Our bounds are essentially the best possible for comparison-based algorithms.\onlyLong{\footnote{Algorithms that use only comparison operations ($\leq, \geq, =$) between elements to distinguish the ordering among them.}}
We also implemented our algorithm and show that it performs well in practice. 
\end{abstract}
\onlyShort{
\copyrightyear{2020}
\acmYear{2020}
\setcopyright{rightsretained}
\acmConference[SPAA '20]{Proceedings of the 32nd ACM Symposium on Parallelism in Algorithms and Architectures}{July 15--17, 2020}{Virtual Event, USA}
\acmBooktitle{Proceedings of the 32nd ACM Symposium on Parallelism in Algorithms and Architectures (SPAA '20), July 15--17, 2020, Virtual Event, USA}\acmDOI{10.1145/3350755.3400268}
\acmISBN{978-1-4503-6935-0/20/07}
}

%

\maketitle


\section{introduction}
\label{sec:intro}

The $K$-nearest neighbors is a well-studied problem
in machine learning with numerous applications. (e.g., \cite{shalev2014understanding}). It
is a non-parametric method used for classification and regression, especially in application such as pattern recognition.
The algorithmic problem is as follows.
We are given a (training) set of $n$ data points ($n$
can be potentially very large and/or each point can be
in a high dimensional space) with labels and a query point $q$. The goal is to assign a label to $q$ based on the labels of the
{\em $K$-nearest points to the query}. \longOnly{Typically, the $n$ points may be in some $d$-dimensional space and we assume that there
is a metric that given two points computes the distance
between the two points (commonly used metrics include Euclidean distance or Hamming distance). }In the {\em classification} problem, one can use the majority
of the labels of the $K$-nearest neighbors to assign
a label to $q$. In the {\em regression} problem, one can assign
the average of the labels (assuming that these are values)
to $q$.

In this paper \onlyShort{(see full version \cite{disknn2020})}, we study distributed algorithms for the $K$-nearest neighbors problem motivated by Big Data and privacy applications. When the data size is very large \longShort{(even storing all points in a single machine might be memory intensive), then distributed computation using multiple machines is helpful. Another even more relevant motivation for distributed computing is that in many instances data is naturally distributed at $k$-sites (e.g., patients data in different hospitals) and it is too costly or undesirable (say for privacy reasons) to transfer all the data to a single location for computing the answer.}{or naturally distributed at $k$-sites (e.g., patients data in different hospitals), then distributed computation using multiple machines is helpful.} 

\longShort{\subsection{Model}}{\noindent{\textbf{Model:}}}
We study the $K$-nearest neighbors problem
in the {\em $k$-machine model}, a model for distributed large-scale data.  (Henceforth, to avoid confusion, between
$K$ and $k$, which are unrelated we
will say $\ell$-nearest neighbors).
The $k$-machine model was introduced in~\cite{klauck2015distributed} and further investigated
in~\cite{chung2015distributed, PanduranganRS16Journal,bandyapadhyay2018near,pandurangan2018distributed}.
The model consists of a set of $k \geq 2$ machines $\{M_1,M_2,\dots,M_k\}$ 
that are pairwise interconnected by bidirectional point-to-point communication links.
Each machine executes an instance of a distributed algorithm. The computation advances
in synchronous rounds where, in each round, machines can exchange messages over their
communication links and perform some local computation. Each link is assumed to have
a bandwidth of $B$ bits per round, i.e., $B$ bits can be
transmitted over each link in each round;  unless otherwise stated, we assume $B = \Theta(\log (n))$. Machines do not share any memory and have no other
means of communication. We assume that each machine has access to a private source of true random bits.

Local computation within a machine is considered to happen instantaneously at zero cost, while the exchange of messages between machines is the costly operation.
However, we note that in all the algorithms of this paper, every machine in every round performs lightweight computations. \onlyLong{In particular, these computations are bounded by (essentially) linear in the size of the input assigned to that machine.}
The goal is to design algorithms that take {\em as few communication rounds as possible}.

\longOnly{
We say that algorithm $\cA$ has \emph{$\epsilon$-error} if, in any run of $\cA$, the output of the machines corresponds to a correct solution with probability at least $1 - \epsilon$.
To quantify the performance of a randomized (Monte Carlo) algorithm $\cA$, we define the \emph{round complexity of $\cA$}  to be the worst-case number of rounds required by any machine when executing $\cA$. 

For the $\ell-$nearest neighbors problem in the $k$-machine model, We assume that the $n$ points are distributed (in an arbitrary fashion) among the $k$ 
machines and the goal is to compute an answer
given a query point in as few rounds as possible. 
We assume that the query point is given to all  machines (or equivalently to a single machine, which can broadcast
to all machines in a round).
}
\longShort{\subsection{The Selection Problem}}{\noindent{\textbf{The Selection Problem:}}}
We note that the $\ell$-nearest neighbors problem really boils down to the {\em selection} problem, where
the goal is to find the $\ell$-smallest value in a 
set of $n$ values. The selection problem has a (somewhat non-trivial) linear time deterministic  algorithm \cite{Cormenbook} as well as simple randomized 
algorithm in the sequential setting. For the
$\ell$-nearest neighbors, one can reduce it to
the selection problem by computing the distance of the query point 
to all the points and then finding the $\ell$-smallest distance among these $n$ distance values. 
All these can be done in $O(n)$ time sequentially.

\longShort{\subsection{Our Results}}{\noindent{\textbf{Our Results:}}}
In this paper, we present efficient bounds
for the $\ell$-nearest neighbors or equivalently
to the $\ell$-selection problem.
Our main result is a randomized algorithm in the $k$-machine model that runs in $O(\log (\ell))$ communication rounds with high probability (regardless of the number of
machines $k$). The message complexity of the algorithm
is also small taking only $O(k\log (\ell))$ messages. Note that if $\ell$ is not very large (which is generally true in practice), then these bounds imply very fast algorithms requiring only a small number of rounds regardless of the number of points and the number of sites (machines). 

Our bounds are essentially the best possible for comparison-based\footnote{We conjecture that the lower bound holds even for non-comparison based algorithms.} algorithms, i.e., algorithms that use only comparison operations ($\leq, \geq, =$) between elements to distinguish the ordering among them. This is due to the existence of a lower bound of $\Omega(\log (n))$ communication rounds (if only one element is allowed to be exchanged per round) for finding the {\em median} of $2n$ elements distributed evenly among two processors \cite{rodeh}. 

One advantage of our distributed algorithm is that a machine exchanges only
{\em distances} between the query point and the points that it contains
and does not send the points themselves. Typically in many applications points have high dimensions and distances don't reveal much about the points. This can be useful in privacy applications. 

We also implement and test our algorithm in a distributed cluster, and show that it performs well compared to a simple algorithm that sends
$\ell$ nearest points from each machine to a single machine
which then computes the answer.
\onlyLong{
In the  simple algorithm each machine locally finds its $\ell$-nearest  points to the query, gathers them on a single machine, and then finds the final $\ell$-nearest points among these $k\ell$ points. Note that this takes
$O(\ell)$ rounds in the $k$-machine model --- exponentially more than our algorithm.}


\onlyLong{
\longShort{\subsection{Related Work}}{\noindent{\textbf{Related Work:}}}
 Methods in \cite{cahsai2017scaling, yang2018efficient} use binary search over the distance of the points from the query point.
 The work of \cite{saukas1998efficient} which is the closest to the spirit of our work, proposed a new distributed algorithm for selection problem aiming to reduce the communication cost. In a model similar to the $k$-machine model (but without explicit bandwidth constraints) they present an algorithm that runs in $O(\log(k\ell))$ rounds\onlyLong{ and $O(k\log(k\ell)\log (\ell))$ message complexity. Their algorithm is deterministic and uses a technique of weighted median}. \longShort{

There are several other works that  investigate applications of $\ell$-nearest neighbors, e.g., see \cite{gao2018demystifying, yu2005monitoring}. Liu et al. in \cite{liu2007clustering} applied $\ell$- nearest neighbors for processing large scale image processing.  Yang et al. \cite{yang2018efficient}  find $\ell$-nearest neighbor objects over moving objects on a large-scale data set. }{There are several other works that  investigate applications of $\ell$-nearest neighbors, e.g., see \cite{gao2018demystifying, yu2005monitoring, yang2018efficient, gao2018demystifying}.}
}

\onlyLong{We remark that in the sequential setting, {\em k-d tree} (short for $k$-dimensional tree) is a well-studied space-partitioning data structure that is used to speed up the processing of nearest neighbor queries \cite{bentley,friedman}.
While k-d tree can help in speeding up computation in the sequential setting, in the $k$-machine model we are concerned only on minimizing the number of communication rounds (and ignoring local computation within a machine). \onlyLong{In the sequential setting, under certain assumptions  k-d tree can give even logarithmic complexity per query point\cite{friedman}.
Here, as far as the round complexity is concerned, this does not matter, since we  can simply send the query point to all machines (takes 1 round)  who then locally compute the distances from the query point to their respective points  and then find the nearest neighbors in 
$O(\log (\ell))$ rounds (does not depend on $k$, the number of machines) using our algorithm. As mentioned earlier, this round complexity is tight in general. }Patwary et al. in \cite{patwary2016panda} used the k-d tree to achieve faster $\ell$-NN calculation in distributed setting.\onlyLong{ They 
implemented a distributed $\ell$-NN based on k-d tree that parallelizes both k-d tree
construction and querying. They created a
large k-d tree for all the points that necessarily involves
global redistribution of points in their k-d tree construction phase.} Since their dimension based redistribution depends on the distribution of input data, their message \onlyLong{and runtime complexity (communication over network) }would be costly.\onlyLong{ Their algorithm would even experience a high round complexity in their construction phase until each node has a non-overlapping subset of input data.}
}

\longShort{\subsection{Definitions}}{\noindent{\textbf{Definitions:}}}
\label{sec:preli}
\onlyShort{We use the notation $dis(p, q)$ to denote the distance between two given points $p$ and $q$ where it can be any absolute norm $||p-q||$.}
\onlyLong{We use the notation $dis(p, q)$ to denote the distance function between two given points $p$ and $q$, where the distance $dis(p, q)$ can be any absolute norm $||p-q||$ distance. Formally, the $\ell$-nearest neighbors problem can be stated as follows.}
\begin{definition}[$\ell$-NN problem]
	Given an input data set $D$, a query data point $q$, and a number $\ell$ while $\ell \leq |D|$, the $\ell$-Nearest Neighbors ($\ell$-NN) problem is finding a set of data points $S$ such that 
	$(S \subset D) \land (|S|=\ell) \land (dis(p_i, q) \leq dis(p_j, q), \forall p_i \in S, p_j \in D\setminus S)$. 
\end{definition}

\section{The Algorithm}
\label{sec:algo}

First we present a distributed algorithm to solve a more general {\em selection} problem:  finding $\ell$-smallest points among $n$ points. Suppose $n$ points are distributed over $k$ machines arbitrarily. The problem is to find the $\ell$-smallest points among those $n$ points. In the end, each machine $i$ outputs a set of points $S_i$ such that $\cup_{i=1}^k S_i$ contains the $\ell$-smallest points. Then we use this algorithm to solve the $\ell$-nearest neighbors problem. For simplicity, let us assume that the points are all distinct; later we explain a simple extension in the algorithm to work for non-distinct points set. To solve this problem we implement the idea of randomized selection \cite{Cormenbook} in the $k$-machine model. 

We point out an implementation issue on the size of the messages used by our algorithm for the nearest neighbors problem. 
For the purpose of analysis, we can assume that each point (or value)
is of size $O(\log (n))$ bits and hence can be sent
through an edge per round in the $k$-machine model. However, for the $\ell$-nearest neighbor problem, points can be high-dimensional and can incur a lot of bits. But it is easy to see that one need not actually transfer points, but only {\em distances} between
the query point to the given (training set) points.
In fact, one can use randomization to choose a unique ID
for each of the $n$ points (choose a random number between say $[1,n^3]$ and they will be unique with high probability). Then one needs to transfer only the ID of the point (of size $O(\log (n))$ bits) and
its corresponding value (distance between the point and the query point)  which we assume can be represented in 
$O(\log (n))$ bits, i.e., all distances are polynomial in $n$.\footnote{We note that if distances are very large, one
can use scaling to work with approximate distances which
will be accurate with good approximation.} Note that
choosing unique IDs also takes care of non-distinct points
as we can use IDs to break ties between points of equal distances.

\subsection{Distributed Selection Algorithm}
This algorithm is a  distributed implementation
of a well-known randomized (sequential) selection algorithm (see e.g.,\cite{Cormenbook}).
The algorithm first elects a leader machine (among the $k$ machines) which propagates the queries and controls the search process. Since the machines have unique IDs, the leader (say, the minimum ID machine) can be elected in a constant number of rounds and $O(\sqrt{k}\log^{3/2}(k))$ messages \cite{kutten2015sublinear}. The leader repeatedly computes a random pivot which partitions the points set into two parts and reduces the search space, i.e., the set of points on which the algorithm executes. Let us now discuss how the leader computes a random pivot and partitions the search space in $O(1)$ rounds.
This constitutes one ``iteration" of the selection algorithm. The leader maintains two boundary variables, namely, $\min$ and $\max$ such that the search points belong to the range $[\min, \max]$. Initially, $\min$ and $\max$ are assigned respectively the minimum (denoted by $\min$) and maximum (denoted by $\max$) value among all the data points. Notice that the leader can get this global minimum and maximum point by asking all the machines their local minimum and maximum in $2$ rounds.

The leader asks the number of points that each machine holds in the range $[\min, \max]$. 
The leader randomly picks a machine $i$ with probability proportional to the number of points a machine holds within the range of $[\min, \, \max]$, i.e., with probability $n_i/\sum_{i=1}^k n_i$, where $n_i$ is the number of points machine $i$ holds in the range. The selected machine $i$ chooses a point $p$ randomly from its set of points in the range $[\min, \max]$. Then it replies back to the leader machine with the pivot $p$.
In the next round, the leader asks the number of points each machine holds within the range $[\min, p]$. Then it gathers all machines' count $n_i$ and accumulates it to $s =\sum_{i=1}^{k} n_i$. If $s=\ell$, it found the correct upper boundary value and terminates the search process. If $s<\ell$, it means the algorithm needs to increase the lower boundary $\min$ to $p$ and adjust the $\ell$ value by subtracting $s$ from $\ell$, i.e., $\ell = \ell - s$. On the other hand, if $s>\ell$, it can discard all the points greater than $p$ by setting $\max$ to $p$. The leader iterates this process until it finds the correct upper boundary. Once the leader finds the correct upper bound ($\max$), it broadcasts a `finished' message with parameter $\max$ so that each machine outputs all the points less than or equal to $\max$ from its input set. 
\onlyLong{The pseudocode is given in Algorithm~\ref{alg:smallest-l-points}.} 
\onlyLong{
\begin{algorithm}
\caption{Finding $\ell$-Smallest-Points}
\label{alg:smallest-l-points}
\begin{flushleft}
\textbf{Input:} $n$ points distributed over $k$ machines (arbitrarily) and $\ell$\\
\textbf{Output:} $\ell$-smallest points among the $n$ data points.
\end{flushleft}  
	\begin{algorithmic}[1]
		\State If there is not a known leader machine $l$ among the $k$ machines, elect one. The leader $l$ runs the following steps.
		\State Leader broadcasts a query message to get the values $(n_i, m_i, M_i)$ from all the machines, where $n_i$ is the no. of points machine $i$ holds, $m_i$ is minimum value and $M_i$ is maximum value among $n_i$ points. 
		\State $\min \gets \min_i\{m_i\}$, $\max \gets \max_i\{M_i\}$, $s\gets \sum_{i=1}^k n_i$ 
		\While{$s > \ell$} \Comment{Each loop runs in  synchronous rounds}
		\label{line:loopb}
		\State Leader selects a random pivot $p$ in the range $[\min, \max]$ by:
		\label{line:pickvalue}
		\begin{enumerate}
			\State Picks a machine $i$ with probability $p_i = \frac{n_i}{s}$ and informs the machine $i$.
			\State Machine $i$ selects a point $p$ uniformly at random from its $n_i$ points and replies back to the leader.
		\end{enumerate}
		\State Leader broadcasts query message $getSize(\min, \, p)$.
			\label{line:bcast1}
		\State Each machine $i$ replies to the leader with $n_i= |\{x\, |\, \min \leq x \leq p\}|$.
			\label{line:reply1}
		\State Leader calculates $s \gets \sum_{i=1}^{k} n_i$
			\label{line:calcnums}
		\If{$s<\ell$}
		    \State $\ell \gets \ell - s$
		    \State $\min \gets p$
		\Else
		    \State $\max \gets p$
		\EndIf
		\EndWhile
		\label{line:loope}
		
		\State Leader broadcasts `finished($\max$)' and each machine outputs all the points satisfying $\{x \, |\, x \leq \max\}$ from its input set 
	\end{algorithmic}
\end{algorithm}
}

\noindent{\textbf{Correctness:}}
In Lemma~\ref{lem:random-pivot}, we show that the leader machine computes the pivot $p$ uniformly at random among all the search points in the range. The algorithm updates boundary values $\min, \max$ and the $\ell$-value according to the randomized selection algorithm. The boundary initialization makes sure that it includes all the data points in the beginning. Thus the algorithm correctly computes the $\ell$-smallest points.
\begin{lemma}\label{lem:random-pivot}
The leader machine \onlyLong{in Algorithm \ref{alg:smallest-l-points} }selects the pivot $p$ uniformly at random from all the points in the range $[\min, \max]$.
\end{lemma}
\onlyLong{\begin{proof}
		Assume there are total $n$ points in the range $[\min, \max]$ distributed over all $k$ machines. The leader selects a machine $i$ with probability $\frac{n_i}{n}$, where $n_i$ is the number of points machine $i$ holds within the range and $n=\sum_{i=1}^{k} n_i$. Now the selected machine $i$ picks the point $p$ randomly among $n_i$ points i.e., with probability $\frac{1}{n_i}$. Therefore the point $p$ (pivot) is selected with probability $\frac{n_i}{n}\cdot \frac{1}{n_i}=\frac{1}{n}$. 
\end{proof}}
Using the above lemma, we show \onlyShort{(in the full paper \cite{disknn2020})} that  the number of elements
in the search process (i.e., in the range  $[\min, \max]$) drops by a constant factor with constant probability.
This implies  that the algorithm stops in $O(\log (n))$ rounds with high probability.
\begin{theorem}\label{thm:smallest-l}
	\longShort{Algorithm \ref{alg:smallest-l-points}}{The above selection Algorithm} computes the $\ell$-smallest points among the $n$ points in the $k$ machine model in $O(\log (n))$ rounds with high probability, and incurs $O(k \log (n))$ messages with high probability, i.e., with probability at least $1-1/n$.
\end{theorem}\onlyLong{
\begin{proof}
The algorithm correctly outputs the $\ell$-smallest points among the $n$ points distributed arbitrarily over $k$ machine model. This can be shown by  a straightforward induction which is similar to the sequential randomized selection algorithm, see e.g., \cite{Cormenbook}. 

Now we show that the algorithm terminates in $O(\log (n))$ rounds with high probability. For the analysis, consider all the points are sorted and placed in an array, although they are in different machines and a machine cannot see the other points. The pivot $p$ is selected uniformly at random from all the points, see  Lemma \ref{lem:random-pivot}. The pivot partitions the set of points into two sets. Let us consider the partition outcomes into good or bad sets. Let the good outcome be that where the pivot is chosen in the middle third of the sorted array, otherwise it's a bad outcome. If the outcome is a good set, then  it discards at least $\frac{1}{3}$ fraction of the points in the range. Thus we define an event $A$ to be a {\em good event} if the randomized partitioning gives good sets, and the complement $\bar{A}$ to be the {\em bad event}. 

   The number of good events cannot be more than $\log_{3/2} (n)$ as each good event keeps at most $\frac{2}{3}\dot n$ points and discards the rest, i.e., all the points will
 be exhausted after $\log_{3/2} (n)$ good events. Since a good event occurs with probability $\frac{1}{3}$, in expectation, an execution path of length $L$ will have $\frac{1}{3}L$ good events. That is, to get $\log_{3/2} (n)$ good events, in expectation the execution path length is at most $3\log_{3/2} (n)$. In other words, in expectation, the algorithm recurs $c\log (n)$ times, where $c$ is a constant such that $3 \log_{3/2} (n) < c \log (n)$. Then applying a standard Chernoff bound \cite{mitzenmacher2017probability}, it can  be shown that the number of iterations cannot be more than $O(\log (n))$ with high probability. Consequently, with $n$ elements in play at the start, the union bound also gives high probability bound on the round complexity of $O(\log (n))$ of the algorithm. 
  
Finally the message complexity of the algorithm is $O(k \log (n))$ with high probability as the leader communicates with all the other machines a constant number of times in a single iteration. Each time the message cost is $O(k)$ through $k-1$ edges from the leader to all the other machines. The massage complexity of leader election algorithm \cite{kutten2015sublinear} is $O(\sqrt{k}\log^{3/2}(k))$. Hence the claimed message complexity bound.      
\end{proof}}

\subsection{Distributed $\ell$-NN Algorithm}\label{sec:fast-algo}
We extend the above algorithm to compute $\ell$-nearest neighbors (or, $\ell$-NN) of a given query point $q$ from a large data set $D$ distributed over the $k$ machines. Assume the machine $i$ gets the set of points $D_i$ as input. We assume that $|D_i| \leq \ell$ for all the machines, since, if a machine $i$ gets more than $\ell$ data points as input, it keeps only $\ell$ points whose distance from $q$ is minimum and discards the rest of the data points. \longShort{This is because a single machine can hold at most all the $\ell$-NN points. Thus a maximum of $k\ell$ input points to be considered to compute $\ell$-NN points. Notice that by applying the Algorithm~\ref{alg:smallest-l-points} directly on these $k\ell$ points one can design an algorithm which computes $\ell$-NN in $O(\log (\ell) + \log (k))$ rounds. In fact, e}{E}ach machine $i$ locally computes the distance $d_{ij} = dis(p_{ij}, q)$ such that all the points $p_{ij} \in D_i$ and maintains the pair $(p_{ij}, \, d_{ij})$. \longShort{Then the system runs the Algorithm~\ref{alg:smallest-l-points}}{Then we apply the selection algorithm} on the distance values $\cup_{i=1}^k d_{ij}$ and \longShort{outputs}{output} the corresponding points $p_{ij}$s. This takes $O(\log (k\ell)) = O(\log (k) + \log (\ell))$ rounds, since the number of candidate points is at most $k\ell$.\footnote{The number of rounds will hold under expectation and with probability guarantee
at least  $1-\frac{1}{k \ell}$, since the number of
points is $k\ell$. If  $\ell$ is very small,
say constant, then even a trivial algorithm of transferring the points to the leader machine will give a small number of rounds.}

We now present a randomized algorithm (Algorithm~\ref{alg:2}) whose running time is $O(\log (\ell))$ rounds, which is independent of the number of machines $k$. 
The main idea of the algorithm is to apply a sampling technique to reduce the 
search space (i.e., candidate points) from (at most) $k\ell$  to $O(\ell)$. Then we apply the \longShort{Algorithm~\ref{alg:smallest-l-points}}{distributed selection algorithm} on these reduced set of candidate points to obtain our main result\onlyShort{(proof in full paper \cite{disknn2020})}. 

\begin{algorithm}[t]
	\caption{Distributed $\ell$-NN Computation}
	\label{alg:2}
	\begin{flushleft}
	\onlyLong{
	\textbf{Input:} $n$ points distributed over $k$ machines arbitrarily, where each machine gets approximately $n/k$ points, query point $q$ and the parameter $\ell$.\\
	\textbf{Output:} $\ell$-nearest neighbors to the query point q.
	}
	\onlyShort{
	\footnotesize{
	\textbf{Input:} Query point $q$, the parameter $\ell$.\\
	\textbf{Output:} $\ell$-nearest neighbors to the query point q.
	}
	}
	\end{flushleft}
	\begin{algorithmic}[1]
	\State Elect a leader machine among $k$ machines (using the leader election algorithm in \cite{kutten2015sublinear}). 
	\State If a machine $i$ has more than $\ell$ data points, it keeps $\ell$ points whose distance from $q$ is minimum and discards other points. Further, if a machine has less than $\ell$ data points, it adds enough sentinel points of $\infty$ (infinity) value to make the size of the data points exactly $\ell$. Let's denote the points set at machine $i$ by $S_i$, where $|S_i| = \ell$. \label{line:makesizel}   
		\State Each machine $i$ samples $12 \log (\ell)$ points randomly and independently from the set $S_i$.
		\State Each machine sends its sampled points to the leader machine. \label{fast:initial-time} 
		\label{fast:line:gather}
		\State Leader sorts these $12k\log (\ell)$ based on their distance from $q$ and stores in an array. Let $r$ be the point at index $21\log(\ell)$ in the sorted array.
		\label{fast:line:p}
		\State Leader broadcasts point $r$.
		\State Each machine $i$ removes any point larger than $r$ and any of added fake data from the set $S_i$. 
		\label{fast:line:prone}
		\State Each machine $i$ computes $d_{ij} = dis(p_{ij}, q)$ for all $p_{ij}\in S_i$ and stores them as $(p_{ij},\, d_{ij})$. 
		\label{fast:line:afterp}
	
	\State The leader machine runs the \longShort{Algorithm~\ref{alg:smallest-l-points}}{distributed selection algorithm} where the input to the algorithm is those $d_{ij}$ points.
	\State Each machine outputs the $p_{ij}$ points corresponding to the output points $d_{ij}$ of the \longShort{Algorithm~\ref{alg:smallest-l-points}}{selection algorithm}.	
	\end{algorithmic}
\end{algorithm}
\onlyLong{
\begin{figure}[H]
	\includegraphics[width=0.7\textwidth]{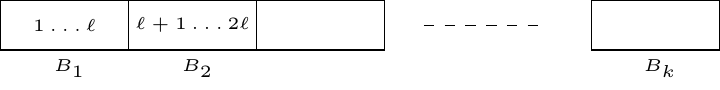}
	\caption{An ascending sorted array $B$ of the $k\ell$ points based on their distances from the query point $q$.}
	\label{fig:sortall}
\end{figure}
}

\longOnly{
\begin{lemma}\label{lem:samplinglogl}
The initial sampling of the Algorithm~\ref{alg:2} reduces the search points to $O(\ell)$ (from $k\ell$ points) with high probability (in $\ell$). 
\end{lemma}
\begin{proof}
Recall that the $n$ data points arbitrarily distributed over $k$ machines. 
We assume that all machines has exactly $\ell$ points after discarding extra points in a machine. Otherwise, if this is not the case, a machine can simply add enough sentinel points (say infinity value), so that all machines have exactly $\ell$ points. (See Step~\ref{line:makesizel} in Algorithm~\ref{alg:2}).   

Now for the sake of analysis, let us assume that all the $k\ell$ points are sorted based on the distances from the query point $q$ and stored in an array $B$. Let the first $\ell$ points in $B$ belong to a block $B_1$, the second $\ell$ points to a block $B_2$, and so on. So there are $k$ blocks; see  Figure~\ref{fig:sortall}.
Let $A$ be the set of sampled $12k\log(\ell)$ points, again consider sorted in ascending order. Let $X_i$ be a random variable denoting the number of these sampled points belonging to the block $B_i$. Since the points are sampled
uniformly at random in each machine, the expected value of $X_i$ is $\mu = E(X_i) = 12 \log(\ell)$. Then by Chernoff bound, $\Pr(X_i \geq (1+\delta) \mu) \leq e^{\frac{-\delta^2 \mu}{3}}$ with $\delta=\sqrt{0.5}$, we get: 
\begin{equation}
    \label{equpper}
    \Pr(X_i \geq 21 \log (\ell))  \leq \Pr(X_i\geq(1+\sqrt{0.5}) 12 \log (\ell)) \leq \frac{1}{\ell^2}.
\end{equation}
Again by Chernoff bound, $\Pr(X_i \leq (1-\delta) \mu) \leq e^{\frac{-\delta^2 \mu}{2}}$ with the same $\delta=\sqrt{0.5}$, we get:
\begin{equation}
    \label{eqlower}
    \Pr(X_i \leq 2 \log (\ell)) \leq \Pr(X  \leq (1-\delta) \mu) \leq \frac{1}{\ell^3}.
\end{equation}
		
Thus, with high probability (in $\ell$) there are at least $2\log (\ell)$ and at most $21\log(\ell)$ sampled points in the block $B_i$ and by a union bound, this holds for the first $O(\log (\ell))$ blocks for a sufficiently large constant. Let $E$ be the event that the selected point $r$ at index $21\log(\ell)$ in the array $A$ (in Step~\ref{fast:line:p} of Algorithm~\ref{alg:2}) belongs to blocks from $B_2$ to $B_{11}$ and not $B_1$ nor beyond $B_{11}$. From Equation~\ref{equpper}, the number of sampled points in block $B_1$, i.e., $X_1$ is less than $21 \log (\ell)$ w.h.p (in $\ell$). Hence, the point $r$ does not belong to $B_1$ w.h.p. Similarly, the point $r$ cannot belong to block $B_i$ for $i > 11$ as by the Equation~\ref{eqlower}, each $X_i > 2 \log (\ell)$ w.h.p (in $\ell$) and $(21 \log (\ell))/(2 \log (\ell)) < 11$. 
So the probability of the complement event of $E$ is:
		$$\Pr(\bar{E}) =  \Pr(r \in B_1) + \sum_{i > 11} \Pr(r \in B_i) \leq  \frac{1}{\ell^2} + \ell * \frac{1}{\ell^3} \leq \frac{2}{\ell^2}.$$
		
		Therefore, $\Pr(E) > 1 - \frac{2}{\ell^2}$. So the selected point $r$ belongs to a block $B_i$, $1 < i \leq 11$ with high probability (in $\ell$). That is the size of candidate points after pruning at Step \ref{fast:line:prone} becomes at most $11\ell$ with high probability (in $\ell$). 
\end{proof}}
\onlyLong{Thus we get the main result.
\begin{theorem}\label{thm:main-fast}
Algorithm~\ref{alg:2} computes $\ell$-NN in $O(\log (\ell))$ rounds and using  $O(k\log (\ell))$ messages
with high probability.
\end{theorem}
\begin{proof}
The leader election takes $O(1)$ rounds. The initial sampling which reduces the size of candidate points to $11\ell$ takes $O(\log (\ell))$ rounds, see Step~\ref{fast:initial-time}. Then it runs the selection algorithm on these $11\ell$ points, which takes $O(\log{(11\ell)}) = O(\log (\ell))$ rounds to compute $\ell$-NN (from Theorem~\ref{thm:smallest-l}). Thus the time complexity is $O(\log (\ell))$ rounds. The message complexity is bounded by  $O(k \log (\ell))$ as both the initial sampling and the selection algorithm incur $O(k\log(\ell))$ messages.   
\end{proof}


}

\onlyShort{
\begin{theorem}\label{thm:main-fast}
Algorithm~\ref{alg:2} computes $\ell$-NN in $O(\log (\ell))$ rounds and uses  $O(k\log (\ell))$ messages
with high probability.
\end{theorem}
}

\section{Experimental Results}\label{sec:experiment}
We ran the Algorithm~\ref{alg:2} using Crill cluster from the University of Houston \footnote{http://pstl.cs.uh.edu/resources/crill-access} which has 16 NLE Systems nodes. Each node has four 2.2 GHz 12-core AMD Opteron processor (48 cores total) and 64 GB main memory. \onlyLong{We used each core as a processing unit for our experiments.}We used a (synthetic) random data set. Each process generated $2^{22}$ random points independently between $0$ and $2^{32}-1$.

We compare the performance of our $\ell$-NN algorithm with the following simple method: each machine finds its local $\ell$-NN. Then it transfers all of them to a leader machine that finds the final $\ell$-NN among those points.
For each simulation, the leader machine chooses a random number  between $0$ and $2^{32}-1$ as the query point. We ran each simulation $30$ times. 
\onlyLong{
Figure \ref{fig:per} shows our algorithm's performance compared to the simple method. We ran it for $k$ ranging between $2$ and $128$ processing units. Also, each resulting point in the figure is the average of $100$ runs of a simulation with a fixed data set and different $q$ query values. The Figure shows that when we increase the number of cores, we gain significant speed up. For example, when using $128$ cores, our algorithm finds $\ell$-NN $80$ times faster than the simple method.

We note that the speed up, measured in wall clock time is due to the fact that as the number of machines increase, the number of points per machine decreases and hence local computation is faster. Thus although, the number of rounds does not depend on the number of machines, in practice (where local computation time also counts), increasing the number of machines increases the speed up.
}
\vspace{-0.9cm}
\onlyLong{
    \begin{figure}[H]
		\includegraphics[scale=0.8]{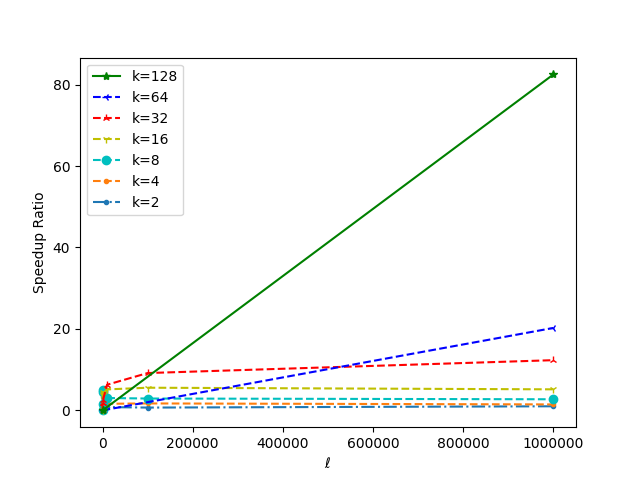}
		\caption{Run-time performance of our Algorithm~\ref{alg:2} compared to the simple method. X-axis shows the number of $\ell$-nearest neighbors w.r.t. a query point and Y-axis shows the execution time ratio of the simple method over our Algorithm~\ref{alg:2}. It shows that the higher the ratio, the higher the algorithm's speedup.}
    	\label{fig:per}
    \end{figure}
}
\onlyShort{
    \begin{figure}[t]
		  \includegraphics[scale=0.5]{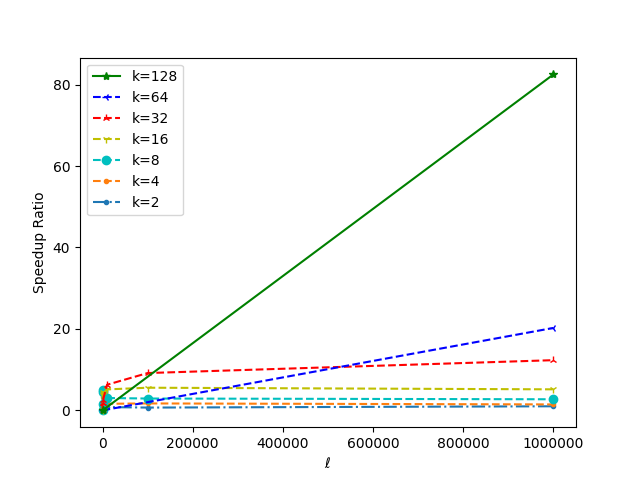}
        	\caption{Run-time performance of our algorithm \ref{alg:2} compared to the simple method. X-axis shows the number of $\ell$-nearest neighbors w.r.t. a query point and Y-axis shows the execution time ratio of the simple method over our algorithm \ref{alg:2}. It shows that the higher the ratio, the higher the algorithm's speedup.}
    	\label{fig:per}
    \end{figure}
}
\onlyLong{
\section{Conclusion}
We studied the well-known $K$-nearest neighbors problem in the distributed $k$-machine model. The $K$-NN problem has numerous applications 
in machine learning and other areas of sciences. Our main contribution is a randomized algorithm which computes the $K$-nearest neighbors with respect to a given query point in $O(\log (K))$ rounds with high probability. The algorithm also uses a small number of messages, incurring only $O(k\log(K))$ messages. We believe that our algorithm can be used as a subroutine for many other problems. It would be interesting to explore other machine learning problems in the $k$-machine model. }

\section*{Acknowledgement} The authors thank Aravind Srinivasan for his valuable comments. 

\balance
\bibliographystyle{abbrv}
\bibliography{./main.bib}


\end{document}